\pdfoutput=1
\newif\ifFull
\Fulltrue
\documentclass{article}
\usepackage{graphicx}
\usepackage{url}
\usepackage{times}
\usepackage{epstopdf}
\usepackage{subfigure}
\ifFull
\usepackage{algorithmic}
\else
\usepackage[noend]{algorithmic}
\fi

\makeatletter
\def\@begintheorem#1#2{\sl \trivlist \item[\hskip \labelsep{\bf #1\ #2:}]}
\def\@opargbegintheorem#1#2#3{\sl \trivlist
      \item[\hskip \labelsep{\bf #1\ #2\ #3:}]}
\makeatother

\newtheorem{theorem}{Theorem}

\ifFull
\newenvironment{proof}{\noindent{\bf Proof:}}{\hspace*{\fill}\rule{6pt}{6pt}\bigskip}
\fi

\begin{document}

\title{Invertible Bloom Lookup Tables}

\author{
{Michael T. Goodrich} \\
Dept.~of Computer Science \\ 
University of California, Irvine 
\ifFull
\\ \url{http://www.ics.uci.edu/~goodrich/}
\fi
\and
Michael Mitzenmacher \\
Dept.~of Computer Science \\ 
Harvard University 
\ifFull
\\ \url{http://www.eecs.harvard.edu/~michaelm/}
\fi
}

\date{}

\maketitle 
\begin{abstract}
We present a version of the Bloom filter data structure that supports
not only the insertion, deletion, and lookup of key-value pairs, but
also allows a complete listing of the pairs it contains with high probability,
as long the number of key-value pairs is below a designed threshold.
Our structure allows the number of key-value pairs to greatly exceed
this threshold during normal operation. Exceeding the threshold simply
temporarily prevents content listing and reduces the probability of a
successful lookup. If entries are later deleted to return the
structure below the threshold, everything again functions
appropriately.  We also show that simple variations of our structure
are robust to certain standard errors, such as the deletion of a key
without a corresponding insertion or the insertion of two distinct
values for a key.  The properties of our structure make it suitable
for several applications, including database and networking
applications that we highlight.
\end{abstract}

\section{Introduction}
The Bloom filter data structure~\cite{bloom1970space} is a well-known
way of probabilistically supporting dynamic set membership queries 
that has been used in a multitude of applications (e.g.,
see~\cite{broder2004network}). The key feature of a standard Bloom
filter is the way it trades off query accuracy for space efficiency,
by using a binary array $T$ (initially all zeroes) and $k$ random hash
functions, $h_1,\ldots,h_k$, to represent a set $S$ by assigning
$T[h_i(x)]=1$ for each $x\in S$.  To check if $x\in S$ one can check
that $T[h_i(x)]=1$ for $1 \leq i \leq k$, with some chance of a false
positive.  This representation of $S$
does not allow one to list out the contents of $S$ given only
$T$. This aspect of Bloom filters is sometimes viewed as a
feature, in settings where 
some degree of privacy protection
is desired
(e.g., see~\cite{bbav-ppidn-09,bc-pesue-07,goh03,nk-csbf-09,pws-pppp-06}).
Still, in many domains one would benefit from a similar set 
representation that would also allow listing out
the set's contents~\cite{eg-sirrt-10}.

In this paper, we are interested not in simply representing a set, but
instead in methods for probabilistically representing a lookup table
(that is, an associative memory) of key-value pairs, where the keys
and values can be represented as fixed-length integers.  Unlike
previous approaches (e.g.,
see~\cite{bonomi2006beyond,chazelle2004bloomier}), we specifically
desire a data structure that supports the listing out of all of its
key-value pairs.  We refer to such a structure as an \emph{invertible
Bloom lookup table} (IBLT).

\subsection{Related Work}
Our work can be seen as an extension of the invertible Bloom
filter data structure of Eppstein and Goodrich~\cite{eg-sirrt-10},
modified to store key-value pairs instead of only keys.  Our
analysis, however, supersedes the analysis of the previous paper in
several respects, in terms of efficiency and tightness of the analysis
(as well as correcting some small deficiencies).
In particular, our analysis demonstrates the natural connection between
these data structures and cores of random hypergraphs, similar to the
connection found previously for cuckoo hashing and erasure-correcting
codes (e.g., see~\cite{dietzfelbinger2009tight,luby2001efficient}).
This provides both a significant constant factor reduction in the required
space for the data structure, as well as an important reduction in the error probability
(to inverse polynomial, from constant in \cite{eg-sirrt-10}).  
In addition, our IBLT supports some usage cases and
applications (discussed later in this section) that are not supported
by a standard invertible Bloom filter.

Given its volume,
reviewing all previous work on Bloom filters is, not surprisingly,
beyond the scope of this paper 
(e.g., see~\cite{bmpsv-iccbf-06,broder2004network,pss-chseb-09} for
some excellent surveys).
Nevertheless,
two closely related works include 
Bloomier filters~\cite{chazelle2004bloomier} and
Approximate Concurrent State Machines (ACSMs)~\cite{bonomi2006beyond}, 
which are
structures to store and track key-value pairs.  
Song {\it et al.}~\cite{sdtl-fhtlu-05} store elements in a Bloom-like hash
table in a fashion that is somewhat similar to that of 
a cuckoo hash table~\cite{dm-chfa-03,naor-history,pr-ch-04}.
Cohen and Matias~\cite{cm-sbf-03} describe an extended Bloom filter that
allows for multiplicity queries on multisets.
While an IBLT is not
intended for changing values, it can be used in such settings 
where a key-value pair can be explicitly deleted and a new value for the key
re-inserted.  Again, an IBLT has additional features, including listing,
graceful handling of data exceeding the listing threshold, and counting multiplicities,
which make it useful for several
applications where these other structures are insufficient.

Another similar structure is the recently developed counter braid
architecture~\cite{lu2008counter}, which keeps an updatable count
field for a set of flows in a compressed form, with the compressed
form arising by a careful use of hashing and allowing reconstruction
of the count for each flow.  Unlike an IBLT, however, the flow list
must be kept explicitly to read out the flow counts, and such lists do
not allow for direct lookups of individual values.  There are several
other differences from our work due to their focus on counters, but
perhaps most notable is their decoding algorithm, which
utilizes belief propagation.  Additional work in the area of
approximate counting of a similar flavor but with very different goals
from the IBLT includes the well-known CM-sketch \cite{CM-sketch} and
recent work by Price \cite{Price}.

\subsection{Our Results}
We present a deceptively simple
variation of the Bloom filter data structure that is
designed for key-value pairs and further avoids the limitation of
previous structures (such
as~\cite{bonomi2006beyond,chazelle2004bloomier}) that do not allow the
listing of contents.  
As mentioned above,
we call our structure an {\em invertible Bloom lookup table}, 
or IBLT for short.
Our IBLT supports insertions, deletions, and lookups in $O(k)$ time,
where $k$ is the number of random hash functions
used (which will typically be a constant in practice).
Just as Bloom filters have false positives,
our lookup operation works only with constant probability, although this
probability can be made quite close to 1.

Our data structure
also allows for a complete listing of the
key-value pairs the structure contains, with high probability, whenever
the current number, $n$, of such pairs lies below a certain 
threshold capacity, $t$, a parameter that is part of the structure's design.
This listing takes $O(t)$ time.  
In addition, because the content-listing operation succeeds with high
probability, one can also use it as a backup in the case that
a standard lookup fails---namely, if a lookup fails, perform a
listing of key-value pairs until one can retrieve 
a value for the desired key.

Our IBLT construction is also space-efficient, 
requiring space\footnote{As in the standard RAM model,
  we assume in this paper that
  keys and values respectively fit in a single word
  of memory (which, in practice, could actually be any fixed number of
  memory words), and we characterize the space used by our data
  structure in terms of the number of memory words it uses.}  
at most linear 
in $t$, the threshold number of keys, even if the number, $n$,
of stored key-value pairs grows well
beyond $t$ (for example, to polynomial in $t$) at points in time.
One could of
course instead keep an actual list of key-value pairs with linear
space, but this would require space linear in $n$, i.e., 
the {\em maximum} number of keys, 
not the target number, $t$, of keys.  Keeping a list also
necessarily requires more computationally expensive lookup operations than our
approach supports.

We further show that with some additional checksums we can tolerate
various natural errors in the system.  For example, we can cope with
key-value pairs being deleted without first being inserted, or keys
being inserted with the same value multiple times, or keys mistakenly
being inserted with multiple values simultaneously.  Interestingly,
together with its contents-listing ability, this error tolerance leads
to a number of applications of the IBLT, which we discuss next.

\subsection{Applications and Usage Cases}
There are a number of possible applications and usage cases for
invertible Bloom lookup tables, some of which we explore here.

\paragraph{Database Reconciliation}
Suppose Alice and Bob hold distinct, but similar, copies, $D_A$ and
$D_B$, of an indexed database, $D$,
and they would like to reconcile the differences between $D_A$ and
$D_B$.
For example, Alice could hold a current version of $D$ and Bob could
hold a backup, or Alice and Bob could represent 
two different copies of someone's
calendar database (say, respectively on a desktop computer and 
a smartphone) that now need to be symmetrically synchronized.
Such usage cases are natural in database systems, 
particularly for systems that take the approach advocated 
in an interesting recent CACM article by Stonebraker~\cite{s-isdc-10}
that argues in favor of
sacrificing consistency for the sake of availability and
partition-tolerance, and then regaining consistency by performing database
reconciliation computations whenever needed.
Incidentally,
in a separate work, 
Eppstein {\it et al.}\footnote{Personal communication.} are
currently empirically exploring a similar application 
of the invertible Bloom table technology
for the reconciliation of two distributed mirrors
of a filesystem (say, in a peer-to-peer network).

To achieve such a reconciliation with low overhead, \hbox{Alice}
constructs an IBLT, $\cal B$,
for $D_A$, using indices as keys and checksums of her records as values.
She then sends the IBLT $\cal B$ to Bob, who then deletes 
index-checksum pairs from $\cal B$ corresponding to all
of his entries in $D_B$.
The remaining key-value pairs corresponding to insertions
without deletions identify records that Alice has that Bob doesn't
have, and the remaining key-value pairs corresponding to deletions
without insertions identify records that Bob has that Alice doesn't have.
In addition, as we show, Bob
can also use $\cal B$ to identify records that they both possess but
with different checksums.
In this way, Alice needs only to send a message $\cal B$
of size $O(t)$, where $t$ here is an upper bound on 
the number of differences between $D_A$ and $D_B$, for Bob 
to determine the identities of their differences (and a symmetric property
holds for a similar message from Bob to Alice).

\paragraph{Tracking Network Acknowledgments}
As another example application, consider a router $\cal R$
that would like to track
TCP sessions passing through $\cal R$.  
In this case, each session corresponds to a
key, and may have an associated value, such as the source or the
destination, that needs to be tracked.  When such flows are initiated in
TCP, particular control messages are passed that can be easily
detected, allowing the router to add the flow to the structure.
Similarly, when a flow terminates, control messages ending the flow
are sent.  The IBLT supports fast insertions and deletions and
can be used to list out the current flows in the system at various
times, as long as the number of flows is less than some preset threshold, $t$.
Note this work can be offloaded simply by sending a copy of the
IBLT to an offline agent if desired.  Furthermore, the
IBLT can return the value associated with a flow when queried,
with constant probability close to 1.  Finally, if at various points,
the number of flows spikes above its standard level to well above $t$,
the IBLT will still be able to list out the flows and perform
lookups with the appropriate probabilities once the total load returns
to $t$ or below.  Again, this is a key feature of the IBLT; all
keys and values can be reconstructed with high probability whenever the
number of keys is below a design threshold, but if the number of keys
{\em temporarily} exceeds this design threshold and later returns to
below this threshold, then the functionality will return at that later
time.  

In this networking setting, sometimes flows do not terminate properly,
leaving them in the data structure when they should disappear.
Similarly, initialization messages may not be properly handled,
leading to a deletion without a corresponding insertion.  We show that
the IBLT can be modified to handle such errors with minimal loss in
performance.  Specifically, we can handle keys that are deleted
without being inserted, or keys that erroneously obtain multiple
values.  Even with such errors, we provide conditions for which all
valid flows can still all be listed with high probability.  Our
experimental results also highlight robustness to these types of
errors.  (Eventually, of course, such problematic keys should be
removed from the data structure.  We do not concern ourselves with
removal policies here; see \cite{bonomi2006beyond} for some
possibilities based on timing structures.)

\paragraph{Oblivious Selection from a Table}
As a final motivating application,
consider a scenario where Alice has outsourced her data storage needs,
including the contents of an important indexed table, $\cal T$, of size $n$,
to a cloud storage server, Bob, because Alice has very limited storage
capacity (e.g., Alice may only have a smartphone).
Moreover, because her data is sensitive and she knows Bob is 
honest-but-curious regarding her data,
she encrypts each
record of $T$ using a secret key, and random nonces,
so that Bob cannot determine the
contents of any record from its encryption alone.  
Such encryptions are not sufficient, however, to fully protect 
the privacy of Alice's data, as
recent attacks show that the way Alice accesses her data
can reveal its contents (e.g., see~\cite{cwwz-sclwa-10}).
Alice needs a way of hiding 
any patterns in the way she accesses her data.

Suppose now that Alice would like to do a simple {\sc select} query on $\cal T$ 
and she is confident that the result will have a size at most $t$,
which is much less than $n$ but
still more than she can store locally.
Thus, she cannot use techniques from private 
information retrieval~\cite{ckgs-pir-98,y-pir-10},
as that would either require storing results back with Bob in a way that 
could reveal selected indices or using yet another server besides Bob.
She could use techniques from recent oblivious RAM 
simulations~\cite{go-spsor-96,gm-mpcho-10,pr-orr-20} to
obfuscate her access patterns, but doing so would require $O(n\log^2 n)$
I/Os. Therefore, using existing 
techniques would be inefficient.

By using an IBLT, on the other hand, she can perfrom her {\sc select} query
much more efficiently.  The advantage comes from the fact that an
insertion in an IBLT accesses a random set of cells (that is, memory
locations) whose addresses depend (via random hash functions) only on
the key of the item being inserted.  Alice thus uses all the indices
for $\cal T$ as keys, one for each record, and accesses memory as
though inserting each record into an IBLT of size $O(t)$.  In fact,
Alice only inserts those records that satisfy her {\sc select} query.
However, since Alice encrypts each write using a secret key and random
nonces, Bob cannot tell when Alice's write operations are actually
changing the records stored in the IBLT, and when a write operation is
simply rewriting the same contents of a cell over again re-encrypted with a
different nonce.  In this way Alice can obliviously 
create an IBLT of size $O(t)$ that contains the result of her query and is
stored by Bob.  Then, using existing methods for 
oblivious RAM simulation~\cite{gm-mpcho-10}, she can 
subsequently obliviously extract the
elements from her IBLT using $O(t\log^2 t)$ I/Os.
With this approach Bob learns nothing about her data from her access pattern.
In addition, the total number of I/Os for her to perform her query is 
$O(n+t\log^2 t)$, which is linear (and optimal) for any $t$ that is $O(n/\log^2 n)$.
We are not currently aware of any other way that Alice can achieve such a
result using a structure other than an IBLT.

\section{A Simple Version of the Invertible Bloom Lookup Table}

In this section, we describe and analyze a simple version of the IBLT.
In the sections that follow we describe how to augment and extend this
simple structure to achieve various additional performance goals.

The IBLT data structure, ${\cal B}$, is a randomized data structure storing
a set of key-value pairs.  It is designed with respect to a
threshold number of keys, $t$; when we say the structure is successful
for an operation with high probability it is under the assumption that
the actual number of keys in the structure at that time, which we henceforth
denote by $n$, is less than or equal to $t$.  
Note that $n$ can exceed $t$ during the course of
normal operation, however.  

As mentioned earlier, we assume throughout
that, as in the standard RAM model,
keys and values respectively fit in a single word
of memory (which, in practice, could actually be any fixed number of
memory words) and that each such word can alternatively be viewed as
an integer, character string, floating-point number, etc.  
Thus, without loss of generality, we view keys and values as
positive integers.

In many cases we take sums of keys and/or values; we must also
consider whether word-value overflow when trying to store these sums in a
memory word.  (That is, the sum is larger than what fits in a data
word.)  As we explain in more detail at appropriate points below, such
considerations have minimal effects.  In most situations, with
suitably sized memory words, overflow may never be a consideration.
Alternatively, if we work in a system that supports graceful
overflows, so that $(x+y)-y = x$ even if the first sum results in an
overflow, our approach works with negligible changes.  Finally, we can
also work modulo some large prime (so that vaues fit within a memory
word) to enforce graceful overflow.  These variations have negligible
effects on the analysis.  However, we point out that in many
settings (except in the case where we may have duplicate copies of the
same key-value pair), we can use XORs in place of sums in our algorithms,
and avoid overflow issues entirely.  

\subsection{Operations Supported}
Our structure supports the following operations:
\begin{itemize}
\item
{\sc insert}$(x,y)$: insert the key-value pair, $(x,y)$, into ${\cal B}$.
This operation always succeeds, assuming that all keys are distinct.
\item
{\sc delete}$(x,y)$: delete the key-value pair, $(x,y)$, from ${\cal B}$.
This operation always succeeds, provided $(x,y)\in {\cal B}$, which
we assume for the rest of this section.
\item
{\sc get}$(x)$: return the value $y$ such that there is a
key-value pair, $(x,y)$, in ${\cal B}$.
If $y=\mbox{null}$ is returned, then $(x,y)\not\in {\cal B}$ for any value of $y$. 
With low (but constant) probability, this operation
may fail, returning a ``not found'' error condition.  In this case there may or may
not be a key-value pair $(x,y)$ in ${\cal B}$.
\item
{\sc listEntries}$()$: list all the key-value pairs being stored in ${\cal B}$.
With low (inverse polynomial in $t$) probability, this operation
may return a partial list along with an ``list-incomplete'' error condition.
\end{itemize}

When an IBLT ${\cal B}$ is first created, it initializes a lookup
table $T$ of $m$ cells.  Each of the cells in $T$ stores a constant
number of fields, each of which corresponds to a single memory word.
We emphasize that an important feature of the data structure is that
at times the number of key-value pairs in ${\cal B}$ can be much
larger than $m$, but the space used for $\cal B$ remains $O(m)$ words.
(We discuss potential issues with word-value overflow where
appropriate.)  The {\sc insert} and {\sc delete} methods never fail,
whereas the {\sc get} and {\sc listEntries} methods, on the other
hand, only guarantee good probabilistic success when $n \leq t$. For
our structures we shall generally have $m = O(t)$, and often we can
give quite tight analyses on the constants required, as we shall see
below.

\subsection{Data Structure Architecture}
Like a standard Bloom filter, an IBLT
uses a set of $k$ random\footnote{We assume, for the sake of
	simplicity in our analysis,
	that the hash functions are fully random, but this does not
	appear strictly required.
	For example, the techniques of \cite{mitzenmacher2008simple}
	can be applied if the data has a
	sufficient amount of entropy.  For worst-case data, we are not
	aware of any results regarding
	the 2-core of a random hypergraph where the vertices for each
	edge are chosen according to
	hash functions with limited independence, which, as we will see,
	would be needed for such a result.
	Similar graph problems with limited independence have recently
	been studied in \cite{alon2008k}.  It
	is an interesting theoretical question to obtain better bounds
	on the randomness needed for our proposed IBLT data structure.
	}
hash functions, $h_1$, $h_2$,
$\ldots$, $h_k$, to determine where key-value pairs are stored.
In our case, each key-value pair, $(x,y)$,
is placed into cells 
$T[h_1(x)]$,
$T[h_2(x)]$,
$\ldots$
$T[h_t(x)]$.
In what follows, for technical reasons\footnote{Incidentally, this
   same technicality can be used to correct a small deficiency in 
   the paper of Eppstein and Goodrich~\cite{eg-sirrt-10}.},
we assume that the hashes
yield distinct locations.  This can be accomplished in various ways,
with one standard approach being to split the $m$ cells into $k$
subtables each of size $m/k$, and having each hash function choose one
cell (uniformly) from each subtable.  Such splitting does not affect
the asymptotic behavior in our analysis and can yield other benefits,
including ease of parallelization of reads and writes into the hash table.
\ifFull
(Another approach would be to select the first $k$ distinct hash values from a
specific sequence of random hash functions.)
\fi

Each cell contains three fields:
\begin{itemize}
\item
a \textsf{count} field,
which counts the number of entries that have been mapped to this cell,
\item
a \textsf{keySum} field,
which is the sum of all the keys that have been mapped to this cell,
\item
a \textsf{valueSum} field,
which is the sum of all the values that have been mapped to this cell.
\end{itemize}

Given these fields, which are initially $0$, performing 
the update operations is fairly straightforward:

\begin{itemize}
\item
{\sc insert}$(x,y)$:
\begin{algorithmic}[100]
\FOR {each (distinct) $h_i(x)$, for $i=1,\ldots,k$}
\STATE
add $1$ to $T[h_i(x)].\mbox{\textsf{count}}$
\STATE
add $x$ to $T[h_i(x)].\mbox{\textsf{keySum}}$
\STATE
add $y$ to $T[h_i(x)].\mbox{\textsf{valueSum}}$
\ENDFOR
\end{algorithmic}
\item
{\sc delete}$(x,y)$:
\begin{algorithmic}[100]
\FOR {each (distinct) $h_i(x)$, for $i=1,\ldots,k$}
\STATE
subtract $1$ from $T[h_i(x)].\mbox{\textsf{count}}$
\STATE
subtract $x$ from $T[h_i(x)].\mbox{\textsf{keySum}}$
\STATE
subtract $y$ from $T[h_i(x)].\mbox{\textsf{valueSum}}$
\ENDFOR
\end{algorithmic}
\end{itemize}

\subsection{Data Lookups}
We perform the {\sc get} operation in a manner similar to how
membership queries are done in a standard Bloom filter. The
details are as follows:
\begin{itemize}
\item
{\sc get}$(x)$:
\begin{algorithmic}[100]
\FOR {each (distinct) $h_i(x)$, for $i=1,\ldots,k$}
\IF {$T[h_i(x)].\mbox{\textsf{count}}\,=\, 0$}
\STATE 
\textbf{return} \mbox{null}
\ELSIF {$T[h_i(x)].\mbox{\textsf{count}}\,=\, 1$} 
\IF {$T[h_i(x)].\mbox{\textsf{keySum}}\,=\, x$}
\STATE 
\textbf{return} {$T[h_i(x)].\mbox{\textsf{valueSum}}$}
\ELSE 
\STATE 
\textbf{return} \mbox{null}
% \ELSIF {$T[h_i(x)].\mbox{\textsf{count}}\,=\, 1$ \textbf{and}
%      $T[h_i(x)].\mbox{\textsf{keySum}}\,=\, x$}
% \STATE 
% \textbf{return} {$T[h_i(x)].\mbox{\textsf{valueSum}}$}
\ENDIF
\ENDIF
\ENDFOR
\STATE \textbf{return} ``not found''
\end{algorithmic}
\end{itemize}

Recall that for now we assume that all insertions and deletions are
done correctly, that is, no insert will be done for an existing key in
$\cal B$ and no delete will be performed for a key-value pair not already
in $\cal B$.  With this assumption, if the above operation 
returns a value $y$ or the \mbox{null} value, then this is the correct
response.  This method may fail, returning ``not found,''
if it can find no cell that $x$ maps to that 
holds only one entry.  Also, as a value is returned only if the count is 1,
overflow of the sum fields is not a concern.  

For a key $x$ that is in $\cal B$, 
consider the probability $p_0$
that each of its hash locations contains no other item.  
Using the standard analysis for Bloom filters 
(e.g., see~\cite{broder2004network}),
we find $p_0$ is:
\[
p_0 = \left(1-\frac{k}{m}\right)^{(n-1)} \approx e^{-kn/m} .
\]
That is, assuming the table is split into $k$ subtables of size $m/k$ (one for each hash
function), each
of the other $n-1$ keys misses the location independently with probability $1-k/m$. 
One nice interpretation of this is that the number of keys that hash to the cell
is approximately a Poisson random variable with mean $kn/m$, and $e^{-kn/m}$
is the corresponding probability a cell is empty.
The probability that a {\sc get} for a key that is in $\cal B$
returns ``not found''
is therefore approximately
\[
(1-p_0)^k \approx
\left(1-e^{-kn/m}\right)^k,
\]
which corresponds to the false-positive rate for a standard Bloom filter.
As is standard for these arguments, these approximations can be readily replaced
by tight concentration results~\cite{broder2004network}.

The probability that a {\sc get} for a key that is {\em not} in $\cal B$ returns``not found''
instead of \mbox{null} can be found similarly.  Here, however, note that every cell hashed to
by that key must be hashed to by at least two other keys from $\cal B$.
This is because an empty cell would lead to a null return value, and
a cell with just one key hashed to it would yield the corresponding true
key value, and hence also lead to a null return value for a key not in $\cal B$.
Using the same approximation -- specifically, that the number of keys from $\cal B$
that land in a cell is approximately distributed as a discrete Poisson random variable with
mean $kn/m$ -- we find this probability is 
\[
\left(1-e^{-kn/m} - \frac{kn}{m}e^{-kn/m} \right)^k.
\]

\subsection{Listing Set Entries}
Let us next consider the method for listing the contents of $\cal B$.
We describe this method in a destructive
fashion---if one wants a non-destructive method, then one should first
create a copy of $\cal B$ as a backup.
\begin{itemize}
\item {\sc listEntries}$()$:
\begin{algorithmic}[100]
\WHILE {there's an $i\in[1,m]$ with $T[i].\mbox{\textsf{count}} = 1$}
\STATE add the pair
$(T[i].\mbox{\textsf{keySum}}\, ,\, T[i].\mbox{\textsf{valueSum}})$ to the output list
\STATE call {\sc delete}($T[i].\mbox{\textsf{keySum}}\, , \, T[i].\mbox{\textsf{valueSum}}$)
\ENDWHILE
\end{algorithmic}
\end{itemize}
It is a fairly straightforward exercise to implement this method in $O(m)$
time, say, by using a link-list-based priority queue of cells in $T$ indexed by
their \textsf{count} fields and modifying the {\sc delete} method to update this
queue each time it deletes an entry from $\cal B$.

If at the end of the while-loop all the entries in $T$ are empty, then we say
that the method {\em succeeded}
and we can confirm that the output list is the
entire set of entries in $\cal B$.
If, on the other hand, there are some cells in $T$ with non-zero counts, then
the method only outputs a partial list of the key-value pairs in $\cal B$.

This process should appear entirely familiar to those who work with
random graphs and hypergraphs.  It is exactly the same procedure 
used to find the 2-core of a random hypergraph (e.g., see~\cite{dietzfelbinger2009tight,molloy2004pure}).  
To make the
connection, think of the cells as being vertices in the hypergraph, and the
key-value pairs as being hyperedges, with the vertices for an edge
corresponding to the hash locations for the key.  The 2-core is
the largest sub-hypergraph that has minimum degree at least 2.  The
standard ``peeling process'' finds the 2-core: while there exists a
vertex with degree 1, delete it and the corresponding hyperedge.  The
equivalence between the peeling process and the scheme for {\sc listEntries} is
immediate.  We note that this peeling process is similarly used for
various erasure-correcting codes, such as Tornado codes and its
derivatives (e.g., see~\cite{luby2001efficient}), 
that have, in some ways, the same flavor as this
construction\footnote{Following this analogy, 
	one could for example, consider {\em irregular}
	versions of the IBLT, where different keys utilize a different number
	of hash values;  
	such a variation could use less space while allowing
        {\sc listEntries} to succeed, or could be used 
	to allow some keys with more hash locations
	to obtain a better likelihood of a successful lookup.  
	These variations are
	straightforward and we do not consider the details further here.}.

Assuming that the cells associated with a key are chosen
uniformly at random, we use known results on 2-cores of random
hypergraphs.  In particular, tight thresholds are known; when the
number of hash values $k$ of each is at least 2, there are constants $c_k >
1$ such that if $m > (c_k +\epsilon)n$ for any constant $\epsilon >
0$, {\sc listEntries} succeeds with high probability, that is with
probability $1-o(1)$.  Similarly, if $m < (c_k - \epsilon)n$
for any constant $\epsilon > 0$, {\sc listEntries} succeeds with probability $o(1)$.
Hence $t = m/c_k$ is (approximately) the design threshold for the IBLT.
As can be found in for example
\cite{dietzfelbinger2009tight,molloy2004pure}, these values are given
by
$$c_k^{-1} = \sup \left \{ \alpha : 0 < \alpha < 1; \forall x \in (0,1), 1 - e^{-k\alpha x^{k-1}} < x \right \}.$$
% An intuitive argument for this value is as follows.

It is easy to check from this definition that $c_k \leq k$, as for
$\alpha = 1/k$ we immediately have $1 - e^{-x^{k-1}} < x$.  
In fact $c_k$ grows much mor slowly with $k$, as shown in 
Table~\ref{tab:thr}, which gives numerical values 
for these thresholds for $3 \leq k \leq 7$. 
Here we are not truly concerned with the
exact values $c_k$; it is enough that only linear space is required.
It is worthwhile noting that $c_k$ is generally close to 1, while to
obtain successful {\sc get} operations we require a number of cells which is
a significant constant factor larger than $n$.  Therefore, in practice
the choice of the size of the IBLT will generally be determined by the
desired probability for a successful {\sc get} operation, not the need for
listing.  (For applications where {\sc get} operations are unimportant and
listing is the key feature, further improvements can be had by using
irregular IBLTs.)
 
\begin{table}[ht]
\begin{center}
\begin{tabular}{c|ccccc}
$k$ & 3 & 4 & 5 & 6 & 7\\\hline
$c_k$ & 1.222 & 1.295 & 1.425  & 1.570  & 1.721 \\
\end{tabular}
\end{center}
\caption{Thresholds for the 2-core rounded to four decimal places.}
\label{tab:thr}
\end{table}

When we design our IBLT, depending on the application, we may
want a target probability for succeeding in listing entries.
Specifically, we may desire failure to occur with probability
$O(t^{-c})$ for a chosen constant $c$ (whenever $n \leq t$).  By choosing $k$ sufficiently
large and $m$ above the 2-core threshold, we can ensure this;
indeed, standard results give that the bottleneck is the possibility
of having two edges with the same collection of vertices, giving a
failure probability of $O(t^{-k+2})$.  The following theorem follows
from previous work but we provide it for completeness.  

\begin{theorem}
As long as $m$ is chosen so that $m > (c_k + \epsilon) t$ for some 
$\epsilon > 0$, {\sc listEntries} fails with
probability $O(t^{-k+2})$ whenever $n \leq t$.
\end{theorem}
\begin{proof}
We describe the result in terms of the 2-core. In what follows we assume $n \leq t$.
The probability that $j$ hyperedges form a non-empty 2-core is dominated
by the probability that these edges utilize only $jk/2$ vertices.  This
probability is at most
\begin{eqnarray*}
{n \choose j}{m \choose jk/2} \left(\frac{jk}{2m} \right )^{jk} & \leq 
& \left ( \frac{n}{e} \right )^j \left ( \frac{2m}{jk} \right )^{jk/2}
\left(\frac{jk}{2m} \right )^{jk} \\
& = &  \frac{n^j}{m^{jk/2}} \left ( \frac{jke}{m} \right)^{jk/2} 
\left (\frac{e}{j} \right)^j. 
\end{eqnarray*}
For $k$ constant, $m > (c_k +\epsilon)n$, and $j \leq \gamma n$ for some constant
$\gamma$, the sum of these probabilities is dominated by the term
where $j=2$, which corresponds to a failure probability of
$O(n^{-k+2})$.  To deal separately with the case of $j > \gamma n$, 
we note that standard analysis of the peeling process shows that, 
as long as $m$ is above the decoding threshold, the probability that
the peeling process fails before reaching a core of size $\delta n$
for any constant $\delta$ is asymptotically exponentially small in $n$.
(See, e.g., \cite{darling2008differential}.)   By this argument the case 
of $j > \gamma n$ adds a vanishing amount to the failure probability, 
completing the proof of the theorem.
\end{proof}

\section{Adding Fault Tolerance to an Invertible Bloom Lookup Table}

For cases where there can be deletions for key-value pairs that are
not already in $\cal B$, or values can be inserted for keys that are
already in $\cal B$, we require some fault tolerance.  We can utilize
a standard approach of adding random checksums to get better fault tolerance.

\paragraph{Extraneous Deletions}
Let us first consider a case with extraneous deletions only.  
Specifically,
we assume a key-value pair might be deleted without a 
corresponding insertion;
however, in this first setting we still assume each 
key is associated with a single
value, and is not inserted or deleted multiple times at any instant.
This causes a variety of problems for both the {\sc get} and {\sc listEntries} routines.
For example, it is possible for a 
cell to have an associated count of 1 even
if more than one key has hashed to it, if there are corresponding
extraneous deletions;  this causes us to re-evaluate our {\sc listEntries}
routine.

To help deal with these issues, we add to our IBLT structure.
We assume that each key $x$ has an additional hash value
given by a hash function $G_1(x)$, which in general we assume
will take on uniform random values in a range $[1,R]$. 
We then require each cell 
has the following additional field:
\begin{itemize}
\item
a \textsf{hashkeySum} field,
which is the sum of the hash values, $G_1(x)$, for all the keys 
that have been mapped to this cell.
\end{itemize}
The \textsf{hashkeySum} field must be of 
sufficiently many bits and the hash function must be 
sufficiently random to make collisions sufficiently unlikely;  
this is not hard to achieve in practice.  
Our insertion and deletion operations must now change accordingly, in
that 
we now must add $G_1(x)$ to each $T[h_i(x)].\mbox{\textsf{hashkeySum}}$ 
on an insertion and subtract $G_1(x)$ during a deletion.   
The pseudocode for these and the other operations is 
given in Figure~\ref{fig:pseudo} at the end of this paper.

The {\textsf{hashkeySum}} field can serve as an extra check. For
example, to check when a cell has a count of 1 that it corresponds to
a cell without extraneous deletions, we check $G_1(x)$
field against the {\textsf{hashkeySum}} field.  For an error to occur,
we must have that a deletion has caused a count of 1 where the count
should be higher, and the hashed key values must align so that their
sum causes a false check.  This probability is clearly at most $1/R$
(using the standard principle of deferred decisions, the ``last hash''
must take on the precise wrong value for a false check).  We will 
generally assume that $R$ is chosen large enough that we can assume a
false match does not occur throughout the lifetime of the data structure,
noting that only $O(\log n)$ bits are needed to handle lifetimes that
are polynomial in $n$.  Notice that even if sum fields overflow, as long
as they overflow gracefully, the probability of a false check is still $1/R$.

% However, it is already somewhat unlikely that an extraneous deletion
% would yield a \textsf{keySum} field that would match a queried key.
% The {\textsf{hashkeySum}} field is more important for the {\sc listEntries}
% algorithm, as a count of 1 is used to recover a key and its associated
% value.  

Let us now consider {\sc get} operations.  The natural approach is to assume
that the {\textsf{hashkeySum}} field will not lead to a false check,
as above.  In this case, on a {\sc get} of a key $x$, if the count field is
0, and the {\textsf{keySum}} and {\textsf{hashkeySum}} are also 0, one
should assume that the cell is in fact empty, and return null.
Similarly, if the count field is 1, and the {\textsf{keySum}} and
{\textsf{hashkeySum}} match $x$ and $G_1(x)$, respectively, 
then one should assume the cell has the
right key, and return its value.  In fact, if the count field is $-1$,
and after negating {\textsf{keySum}} and {\textsf{hashkeySum}} the
values match $x$ and $G_1(x)$, respectively, 
one should assume the cell has the right key, except
that it has been deleted instead of inserted!  In our pseudocode we
return the value, although one could also flag it as an extraneous
deletion as well.  Note, however, that we can no longer return null if
the count field is 1 but the {\textsf{keySum}} field does not match
$x$; in this case, there could be, for example, an additional key
inserted and an additional key extraneously deleted from that cell,
which would cause the field to not match even if $x$ was hashed to
that cell.  If we let $n$ be the number of keys either inserted or
extraneously deleted in the IBLT, then this reduces the probability of
returning null for a key not in ${\cal B}$ to
$\left(1-e^{-kn/m}\right)^k$.  That is, to return null we must have at least
one cell with zero key-value pairs from ${\cal B}$ hashing to it, which occurs (approximately)
with the given probability (using our Poisson approximation).  

% {\bf MM:  Are there other variants I'm missing here?  I feel like
% there might be, although this is the obvious one.}

For the {\sc listEntries} operation, we again use the {\textsf{hashkeySum}}
field to check when a cell has a count of 1 that it corresponds to a cell
without extraneous deletions.  Note here that an error in this check
will cause the entire listing operation to fail, so the
probability of a false check should be made quite low---certainly
inverse polynomial in $n$.  Also note, again, that we can make
progress in recovering keys with cells with a count of $-1$ as well, if the
cell contains only one extraneously deleted key and no inserted keys.
That is,
if a cell contains a count of $-1$, we can negate the
\textsf{count}, \textsf{keySum}, and {\textsf{hashkeySum}} fields, check
the hash value against the key to prevent a false match, and if that
check passes recover the key and remove it (in this case, add it back
in) to the other associated cells.  Hence, a cell cannot yield a key
during the listing process only if more than 
one key, either inserted or deleted,
has hashed to that cell. This is now exactly the same setting as in
the original case of no extraneous deletions, and hence (assuming
that no false checks occur!) the same analysis applies, with $n$
representing the number of keys either inserted or extraneously
deleted.
We give the revised pseudo-code descriptions in Figure~\ref{fig:pseudo}.

\paragraph{Multiple Values}
A more challenging case for fault tolerance occurs when a key can be
inserted multiple times with different values, or inserted and deleted
with different values.  If a key is inserted multiple times with
different values, not only can that key not be recovered, but every
cell associated with that key has been {\em poisoned}, in that it will
not be useful for listing keys, as it cannot have a count of 1 even as
other values are recovered.  (A later deletion of a key-value pair
could correct this problem, of course, but the cell is poisoned at the
time.)  The same is true if a key is inserted and deleted
with different values, and 
here the problem is potentially even worse: if a single
other key hashes to that cell, the count may be 1 and the
\textsf{keySum} and {\textsf{hashkeySum}} fields will be correct even
though the \textsf{valueSum} field will not match the other key's
value, causing errors.

Correspondingly, we introduce an additional check 
for the sum of the values at a cell,
using a hash function $G_2(y)$ for the values, and adding the following field:
\begin{itemize}
\item
a {\textsf{hashvalueSum}} field,
which is the sum of the hash values $G_2(y)$ for all 
the values that have been mapped to this cell.
\end{itemize}
One can then check that the hash of the \textsf{keySum} and \textsf{valueSum} 
take on the appropriate values when
the \textsf{count} field of a cell is 1 (or $-1$) in order 
to see if listing the key-value pair
is appropriate. 

The question remains whether the poisoned cells will prevent recovery
of key values.  Here we modify the goal of {\sc listEntries} to
return all key-value pairs for all valid keys with high
probability---that is, all keys with a single associated value at that
time.  We first claim that if the invalid keys make up a constant
fraction of the $n$ keys that this is not possible under our construction with
linear space.  A constant fraction of the cells would then be
poisoned, and with constant probability each valid key would then hash
solely to poisoned cells, in which case the key could not be
recovered.  

However, it is useful to consider these probabilities, as in practical
settings these quantities will determine the probability of failure.  
For example, suppose $\gamma n$ keys are invalid for some constant $\gamma$.
By our previous analysis, the fraction of cells that are poisoned is concentrated
around $\left(1-e^{-k\gamma n/m}\right)$, and hence the probability that any specific valid
key has all of its cells poisoned is $\left(1-e^{-k\gamma n/m}\right)^k$.  
(While there are other possible ways a key could not be recovered, for example
if two keys have all but one of their cells poisoned and their remaining cell is
the same, this gives a good first approximation for reasonable values, as other
terms will generally be lower order when these probabilities are small.)  
For example, in a configuration we use in our experiments below,
we choose $k = 5$, $m/n = 8$, and $\gamma = 1/10$;  in this case, the probability
of a specific valid key being unrecoverable is approximately $8.16 \cdot 10^{-7}$,
which may be quite suitable for practice.  

For a more theoretical asymptotic analysis, suppose instead that there
are only $n^{1-\beta}$ invalid keys.  Then if each key uses at
least $\lceil 1/\beta \rceil + 1$ hash functions, with high
probability every valid key will have at least one hash location that
does not coincide with a invalid cell.  This alone does not guarantee
reconstructing the valid key-value pairs, but we can extend this idea
to show {\sc listEntries} can provably successfully obtain all valid keys
even with $n^{1-\beta}$ invalid keys; by using $k = \lceil
1/\beta \rceil + 4$ hash functions, we can guarantee with high
probability that every valid key has at least 3 hash locations without
an invalid cell.  (One can raise the probability to any inverse
polynomial in $n$ by changing the constant 4 as desired.)  Indeed, we
can determine the induced distribution for the number of neighbors that are
unpoisoned cells for the valid keys, but the important fact is that
the number of keys with $k$ hashes to unpoisoned cells in $n -
o(n)$.  It follows from the standard analysis of random cores (e.g.,
see Molloy~\cite{molloy2004pure}) that the same threshold as for the
original setting with $k$ hash functions applies.  Hence the number of
cells needed will again be linear in $n$ (with the coefficient
dependent on $\beta$) in order to guarantee successful listing of
keys with high probability.  This yields the following theorem:
\begin{theorem}
Suppose there are $n^{1-\beta}$ invalid keys.   
Let $k = \lceil 1/\beta \rceil + 4$.  Then if  
$m > (c_k + \epsilon) n$ for some $\epsilon > 0$, {\sc listEntries} succeeds with high probability.
\end{theorem}

While this asymptotic analysis provides some useful insights, namely that
full recovery is practical, in practice we expect the analysis above based
on setting $\gamma$ so that the number of invalid keys is $\gamma n$ will
prove more useful for parameter setting and predicting performance.

\paragraph{Extensions to Duplicates}
Interestingly, using the same approach as for extraneous deletions,
our IBLT can handle the setting where the {\em same} key-value pair is inserted multiple
times.
Essentially, this means the IBLT is robust to duplicates, or can
also be used to obtain a count for key-value pairs that are inserted
multiple times. 
We again use
the additional {\textsf{hashkeySum}} and {\textsf{valueSum}} fields.
When the \textsf{count} field is $j$ for the cell, we
take the {\textsf{keySum}}, {\textsf{hashkeySum}}, and
{\textsf{valueSum}} fields and divide them by $j$ to obtain the
proposed key, value, and corresponding hash.  
(Here, note we cannot use XORs in place of sums in our algorithms.)
If the key hash matches,
we assume that we have found the right key and return the 
lookup value or list the key-value pair accordingly, depending on whether a {\sc get} or {\sc listEntries}
operation is being performed.  If it is possible to have the same key appear
with multiple values, as above, then we must {\em also} make use of the
{\textsf{hashvalueSum}} fields, dividing it by $j$ and using it
to check that the value is correct as well.
For the listing operation, the IBLT
deletes $j$ copies of the key-value pair from the other cells.\footnote{Note that here we are making use of the assumption that the hash locations are distinct for a key;  otherwise, the count for the number of copies at this location might not match the number of copies of the key in all the other locations.}  The point here is that a key that appears
multiple times, just as a key that is deleted rather than inserted,
can be handled without significant change to the decoding approach.

The one potential issue with duplicate key-value pairs is in the case of 
word-value overflow
for the memory locations containing the sum; in case of overflow, it
may be that one does not detect that the key hash matches (and
similarly for the {\textsf{hashvalueSum}} fields).  In practice this
may limit the number of duplicates that can be tolerated; however, for
small numbers of duplicates and suitably sized memory fields, overflow 
will be a rare occurrence (that would require large numbers of keys to
hash to the same cell, a provably probabilistically unlikely event).

\paragraph{Fault Tolerance to Lost Memory Subblocks}

We offer one additional way this structure proves resilient to various
possible faults.  Suppose that the structure is indeed set up with $k$
different memory subblocks, one for each hash function.  Conceivably,
one could lose an entire subblock, and still be able to recover all
the keys in a listing with high probability, with only a reduction in
the success probability of a {\sc get} (as long as $k-1$ hashes with a smaller range
space remains sufficient for listing).  In some sense, this is because the
system is arguably overdesigned; obtaining high lookup success
probability when $n$ is less than the threshold $t$ requires a large number
of empty cells, and this space is far more than is needed for decoding.  
% This leads to a space-saving measure, which we discuss next.

\paragraph{An Example Application}

As an example of where these ideas might be used, we return to our
mirror site application.  An IBLT $\cal B$ from Alice can be used by
Bob to find filename-checksum (key-value) pairs where his filename has
a different checksum than Alice's.  After deleting all his key-value
pairs, he lists out the contents of $\cal B$ to find files that he or
Alice has that the other does not. The IBLT might not be empty at this
point, however, as the listing process might not have been able to
complete due to {\em poisoned} cells, where deletions were done for
keys with values different than Alice's values.  To discover these,
Bob can re-insert each of his key-value pairs, in turn, to find any
that may unpoison a cell in $\cal B$ (where he immediately deletes
ones that don't lead to a new unpoisoned cell).  If a new unpoisoned
cell is found found (using the $G_1$ hash function as a check),
then Bob can then remove a key-value pair with the same key as his but
with a different value (that is, with Alice's value).  Note Bob
may then also be able to possibly perform more listings of keys that might
have been previously unrecovered because of the poisoned cells.
Repeating this will discover with high probability all the key-value pairs 
where Alice and Bob differ.

\section{Space-Saving for an Invertible Bloom Lookup Table}

Up to this point, we have not been concerned with minimizing space for
our IBLT structure, noting that it can be done in linear space.  
Nevertheless, there
are a variety of techniques available for reducing the space required,
generally at the expense of additional computation and shuffling of
memory, while still keeping constant amortized or worst-case time
bounds on various operations.  Whether such efforts are worthwhile
depends on the setting; in some applications space may be
the overriding issue, while in others speed or even simplicity might
be more relevant.  We briefly point to some of the previous work that
can offer further insights.  

IBLTs, like other Bloom filter structures, often have a great deal of
wasted space corresponding to zero entries that can be compressed or
fixed-length space required for fields like the \textsf{count} field
that can be made variable-length depending on the value.  This wasted
space can be compressed away using techniques for keeping compressed
forms of arrays, including those for storing arrays of variable-length
strings.  Such mechanisms are explored in for example
\cite{blandford2008compact,pagh2005optimal,raman-succinct} and can be applied here.

A simpler, standard approach (saving less space) is to use
quotienting, whereby the hash value for a key determines a bucket and
an additional quotient value that can be stored.  Quotienting can
naturally be used with the IBLT to reduce the space used for storing
for example the \textsf{keySum} or 
the \textsf{hashkeySum} values.  Also, as previously mentioned, 
in settings without multiple copies of the same key-value pair,
we can use XORs in place of sums to save space.

Finally, we recall that the space requirements arise because of the desire
for high accuracy for {\sc get} operations, not because of the {\sc listEntries} operation.  If one is
willing to give up lookup accuracy---which may be the case if, for example,
one expects the system to be overloaded much of the time---then less space
is needed to maintain successful listing.  

\section{Simulations and Experiments}

We have run a number of simulations to test the IBLT structure and our
analysis.  In these experiments we have not focused on running time; a
practical implementation could require significant optimization.
Also, we have not concerned ourselves with issues of word-value
overflow.  Because of this, there is no need to simulate the data
structure becoming overloaded and then deleting key-value pairs, as
the state after deletions would be determined entirely by the
key-value pairs in the system.  Instead, we focus on the success
probability of the listing of keys and, to a lesser extent, on the
success probability for a {\sc get} operation.  Overall, we have found
that the IBLT works quite effectively and the performance matches our
theoretical analysis.  We provide a few example results.  In all of
the experiments here, we have chosen to use five hash functions.

\begin{figure*}[ht]
  \begin{center}
    \subfigure[{10,000} keys]{\label{fig:plot1}\includegraphics[scale=0.60]{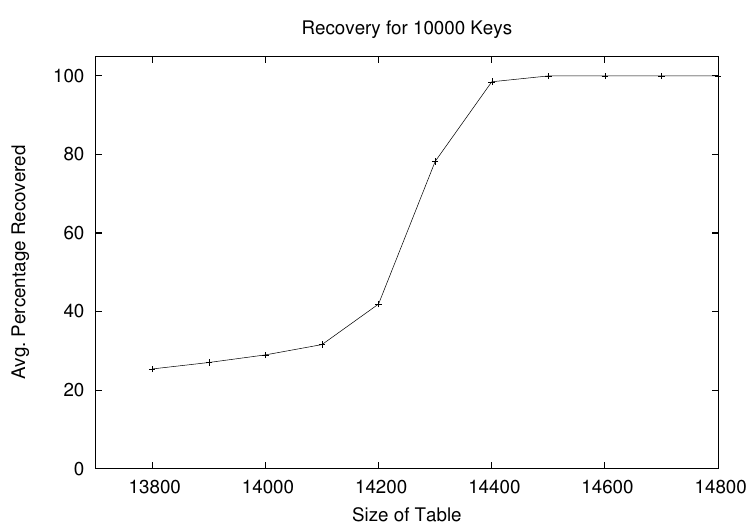}}
    \subfigure[{100,000} keys]{\label{fig:plot2}\includegraphics[scale=0.60]{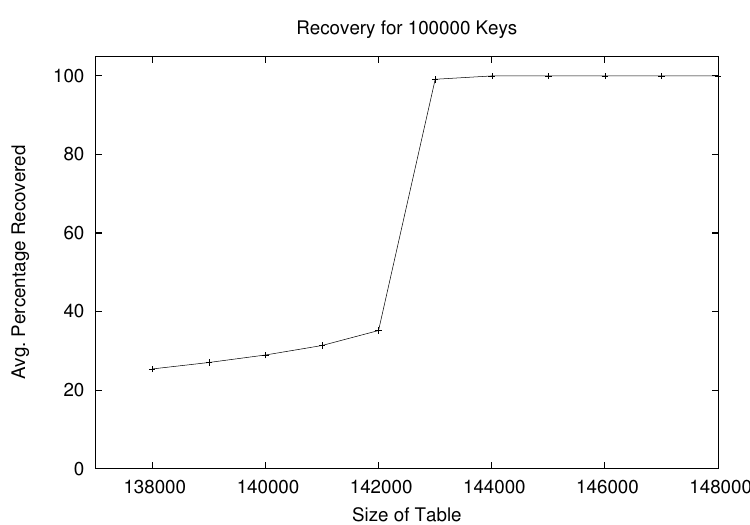}} 
  \end{center}
  \caption{Percentage of key-value pairs recovered around the threshold.  Slightly over the theoretical asymptotic
threshold, we obtain full recovery of all key-value pairs with {\sc listEntries} on all simulations.  Each data point
represents the average of {20,000} simulations.}
  \label{fig:plotx}
\end{figure*}

First, our calculated asymptotic thresholds for decoding from
Table~\ref{tab:thr} are quite accurate even for reasonably small
values.  For example, in the setting where there are no duplicate keys
or extraneous deletions of keys, we repeatedly performed {20,000}
simulations with {10,000} keys, and varied the number of cells.
Table~\ref{tab:thr} suggests an asymptotic threshold for listing all
entries near {14,250}.  As shown in Figure~\ref{fig:plot1}, around
this point we see a dramatic increase in the average number of key-value pairs
recovered when performing our {\sc listEntries} operation.  In fact, at
{14,500} cells only two trials failed to recover all key-value pairs,
and with {14,600} cells or more all trials successfully recover all
key-value pairs.  We performed an additional {200,000} trials with
{14,600} cells, and again all trials succeeded.  In
Figure~\ref{fig:plot2}, we consider {20,000} simulations with
{100,000} keys, where the corresponding threshold should be near
{142,500}.  With more keys, we expect tighter concentration around
the threshold, and indeed with {144,000} cells or more all trials
successfully recover all key-value pairs.  We performed an additional
{200,000} trials with with {144,000} cells, and again all trials
succeeded.

We acknowledge that more simulations would be required to obtain
detailed bounds on the probability of failure to recover all key-value
pairs for specific values of the number of key-value pairs and cells.  This is
equivalent to the well-studied problem of ``finite-length analysis''
for related families of error-correcting codes.  Dynamic programming
techniques, as discussed in \cite{KLS} and subsequent follow-on work,
can be applied to obtain such bounds.  

Our next tests of the IBLT allow duplicate keys with the same value and extraneous
deletions, but without keys with multiple values.  Our analysis
suggests this should work exactly as with no duplicate or extraneous
deletions, and our simulations verify this.  In these simulations, we
had each key result in a duplicate with probability $1/5$, and
each key result in a deletion in place of an insertion with
probability $1/5$.  Using a check on key and value fields, in {20,000}
simulations with {10,000} keys, {80,000} cells, and 5 hash functions,
a complete listing was obtained every time, and {\sc get} operations were
successful on average for {97.83} percent of the keys, matching the
standard analysis for a Bloom filter.  Results were similar with
{20,000} runs with {100,000} keys and {800,000} cells, again with
complete recovery each time and {\sc get} operations successful on average for
{97.83} percent of the keys.

\begin{figure*}[ht]
  \begin{center}
    \subfigure[{10,000} keys]{\label{fig:plot3}\includegraphics[scale=0.60]{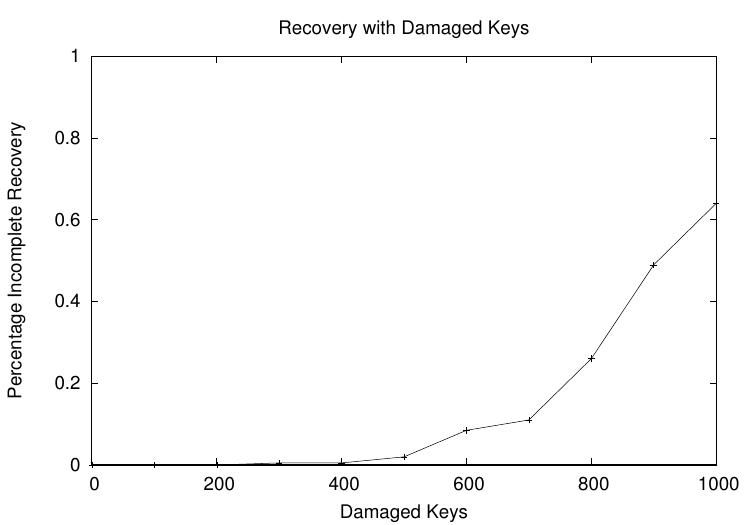}}
    \subfigure[{100,000} keys]{\label{fig:plot4}\includegraphics[scale=0.60]{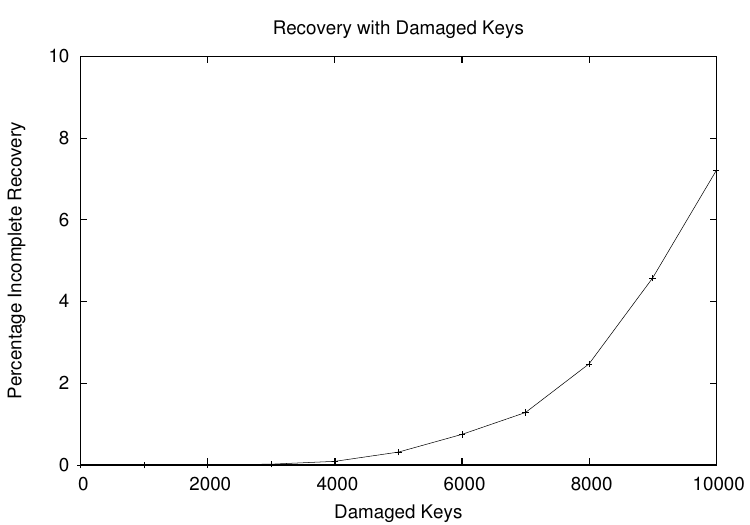}} 
  \end{center}
  \caption{Percentage of trials with incomplete recovery with ``damaged'' keys that have multiple values. Each data point
represents the average of {20,000} simulations.}
  \label{fig:ploty}
\end{figure*}

Finally, we tested the IBLT with keys that erroneously obtain multiple
values.  As expected, these keys can prevent recovery of other
key-value pairs during listing, but do not impact the success probability of {\sc get}
operations for other keys.  For example, again using a check on key
and value fields, in {20,000} simulations with {10,000} keys of which
500 had multiple values, {80,000} cells, and 5 hash functions, the
9500 remaining key-value pairs were recovered {19,996} times; the
remaining 4 times all but one of the 9500 key-value pairs was recovered.  With
1000 keys with multiple values, the remaining key-value pairs were
recovered {19,872} times, and again the remaining 128 times all but
one of the 9000 key-value pairs was recovered.  The average success rate for {\sc get}
operations remained {97.83} percent on the valid keys in both cases, as would be
expected.  We note that with {10,000} keys with {1,000} with multiple
values, our previous back-of-the-envelope calculation showed that each
valid key would fail with probability roughly $8.16 \cdot 10^{-7}$;
hence, with {9,000} other keys, assuming independence we would
estimate the probability of complete recovery at approximately
$0.9927$, closely matching our experimental results.  More detailed
results are give in Figures~\ref{fig:plot3} and~\ref{fig:plot4}, where
we vary the number keys with multiple values for two settings: 
{10,000} keys and {80,000} cells, and {100,000} keys and {800,000} cells.
The results are based on {20,000} trials.  As can be seen complete 
recovery is possible with large numbers of multiple-valued keys in both
cases, but naturally the probability of complete recovery becomes worse with
larger numbers of keys even if the percentage of invalid keys is the same.

We emphasize that even when complete recovery does not occur in this
setting, generally almost all keys with a single value can be
recovered.  For example, in Table~\ref{keyexp} we consider three
experiments.  The first is for {10,000} keys, {80,000} cells, and
{1,000} keys with duplicate values.  The second is the same but with
{2,000} keys with duplicate values.  The third is for {100,000} keys,
{800,000} cells, and {10,000} keys with duplicate values.  Over all
{20,000} trials for each experiment, in no case were more than 3 valid
keys unrecovered.  The main point of Table~\ref{keyexp} is that with
suitable design parameters, even when complete recovery is not
possible because of invalid keys, the degradation is minor.
We suspect this level of robustness may be useful for 
some applications where almost-complete recovery is acceptable.

\begin{table}[ht]
\begin{center}
\begin{tabular}{c|cccc|}
           & \multicolumn{4}{|c|}{Unrecovered Keys} \\ \hline
           & 0  & 1 & 2 & 3 \\\hline
Experiment 1 & 99.360 & 0.640 & 0.000 & 0.000 \\ \hline
Experiment 2 & 83.505 & 14.885  & 1.520 & 0.090 \\ \hline
Experiment 3 & 92.800 & 6.915 & 0.265 & 0.020 \\ \hline
\end{tabular}
\end{center}
\caption{Percentage of trials where 1, 2, and 3 keys are left unrecovered.}
\label{keyexp}
\end{table}

\section{Conclusion and Future Work}
We have given an extension to the Bloom filter data structure to key-value
pairs and the ability to list out its contents.
This structure is deceptively simple, but is able to achieve functionalities
and efficiencies that appear to be unique in many respects, based on our
analysis derived from recent results on 2-cores in hypergraphs.
One possible direction for future work includes whether one 
can easily include methods for 
allowing for multiple values as a natural condition instead of an error.

\subsection*{Acknowledgments}
Michael Goodrich was supported in part by the National Science
Foundation under grants 0724806, 0713046, 0847968, and 0953071.
Michael Mitzenmacher 
was supported in part by the National Science
Foundation under grants IIS-0964473, CCF-0915922 and CNS-0721491, and 
in part by grants from Yahoo! Research, Cisco, Inc., and Google.

{\raggedright
\bibliographystyle{abbrv}
\bibliography{bloom,morebibentries,crypto,cuckoo,cuckoo2}
}

\begin{figure*}[p]
{
\begin{itemize}
\item
{\sc insert}$(x,y)$:
\begin{algorithmic}[100]
\FOR {each $h_i(x)$ value, for $i=1,\ldots,k$}
\STATE
add $1$ to $T[h_i(x)].\mbox{\textsf{count}}$
\STATE
add $x$ to $T[h_i(x)].\mbox{\textsf{keySum}}$
\STATE
add $y$ to $T[h_i(x)].\mbox{\textsf{valueSum}}$
\STATE
add $G_1(x)$ to $T[h_i(x)].\mbox{\textsf{hashkeySum}}$
\ENDFOR
\end{algorithmic}
\item
{\sc delete}$(x,y)$:
\begin{algorithmic}[100]
\FOR {each $h_i(x)$ value, for $i=1,\ldots,k$}
\STATE
subtract $1$ from $T[h_i(x)].\mbox{\textsf{count}}$
\STATE
subtract $x$ from $T[h_i(x)].\mbox{\textsf{keySum}}$
\STATE
subtract $y$ from $T[h_i(x)].\mbox{\textsf{valueSum}}$
\STATE
subtract $G_1(x)$ to $T[h_i(x)].\mbox{\textsf{hashkeySum}}$
\ENDFOR
\end{algorithmic}
\item
{\sc get}$(x)$:
\begin{algorithmic}[100]
\FOR {each $h_i(x)$ value, for $i=1,\ldots,k$}
\IF {$T[h_i(x)].\mbox{\textsf{count}}= 0$ \textbf{and}
      $T[h_i(x)].\mbox{\textsf{keySum}}= 0$  \textbf{and}
      $T[h_i(x)].\mbox{\textsf{hashkeySum}}= 0$}
\STATE 
\textbf{return} \mbox{null}
\ELSIF {$T[h_i(x)].\mbox{\textsf{count}}= 1$ \textbf{and}
      $T[h_i(x)].\mbox{\textsf{keySum}}= x$  \textbf{and}
      $T[h_i(x)].\mbox{\textsf{hashkeySum}}= G_1(x)$ }
\STATE 
\textbf{return} {$T[h_i(x)].\mbox{\textsf{valueSum}}$}
\ELSIF {$T[h_i(x)].\mbox{\textsf{count}}= -1$ \textbf{and}
      $T[h_i(x)].\mbox{\textsf{keySum}}= -x$ \textbf{and}
      $T[h_i(x)].\mbox{\textsf{hashkeySum}}= -G_1(x)$ }
\STATE 
\textbf{return} {$-T[h_i(x)].\mbox{\textsf{valueSum}}$}
\ENDIF
\ENDFOR
\STATE \textbf{return} ``not found''
\end{algorithmic}
\item {\sc listEntries}$()$:
\begin{algorithmic}[100]
\WHILE {there is an $i\in[1,m]$ such that $T[i].\mbox{\textsf{count}} = 1$ or $T[i].\mbox{\textsf{count}} = -1$}
\IF {$T[h_i(x)].\mbox{\textsf{count}}\,=\, 1$ \textbf{and} $T[h_i(x)].\mbox{\textsf{hashkeySum}}\,=\, G_1(T[h_i(x)].\mbox{\textsf{keySum}})$ }
\STATE add the pair,
$(T[i].\mbox{\textsf{keySum}}\, ,\, T[i].\mbox{\textsf{valueSum}})$,
to the output list
\STATE call {\sc delete}($T[i].\mbox{\textsf{keySum}}$,\,$T[i].\mbox{\textsf{valueSum}}$)
\ELSIF {$T[h_i(x)].\mbox{\textsf{count}}\,=\, -1$ \textbf{and} $-T[h_i(x)].\mbox{\textsf{hashkeySum}}\,=\, G_1(-T[h_i(x)].\mbox{\textsf{keySum}})$ }
\STATE add the pair,
$(-T[i].\mbox{\textsf{keySum}}\, ,\, -T[i].\mbox{\textsf{valueSum}})$,
to the output list
\STATE call {\sc insert}($-T[i].\mbox{\textsf{keySum}}$,\,$-T[i].\mbox{\textsf{valueSum}}$)
\ENDIF
\ENDWHILE
\end{algorithmic}
\end{itemize}
}
\caption{Revised pseudo-code for tolerating extraneous deletions.}
\label{fig:pseudo}
\end{figure*}

\end{document}